\documentclass{article}
\usepackage[autostyle, english = american]{csquotes}
\usepackage[english]{babel}
\usepackage{algpseudocodex}
\usepackage{algorithm}
\usepackage{array}
\usepackage{tikz}
\usepackage{todonotes}
\usepackage{amsfonts}
\usepackage{fontawesome5}

\usepackage{subcaption}

\newcommand{\opt}{\mathsf{opt}}
\newcommand{\cif}{C_{\mathsf{if}}}
\newcommand{\celse}{C_{\mathsf{else}}}
\newcommand{\SC}{\mathsf{SC}}
\newcommand{\CSC}{C_{\mathsf{\SC}}}
\newcommand{\CSCsets}{C_{\mathsf{\SC,sets}}}
\newcommand{\CSCpen}{C_{\mathsf{\SC,pen}}}
\newcommand{\CNIfac}{\CSCsets}
\newcommand{\CNIcon}{\CSCpen}

\newcommand{\pcsc}{\mathsf{PC SC}}
\newcommand{\cost}{\mathsf{cost}}

\newcommand{\nwsf}{\mathsf{NWSF}}
\newcommand{\adv}{\mathsf{Adv}}
\newcommand{\Boundary}{\bd B}
\newcommand{\greedy}{\mathsf{greedy}}
\newcommand{\Conf}{\mathsf{Con}}
\renewcommand{\S}{\mathcal{S}}
\newcommand{\mcost}{\mathsf{cost}}
\newcommand{\task}{\mathsf{task}}
\newcommand{\TaskSpace}{\mathcal{T}}
\newcommand{\guess}{\mathsf{guess}}
\newcommand{\ktild}{\tilde{k}}

\newcommand{\Rsc}{2^{j-5}}
\newcommand{\Csc}{2^{j-6}}
\newcommand{\Rinc}{2^{j-5}}
\newcommand{\Rfac}{2^{j-3}}

\newcommand{\Rscd}{R_\mathsf{sc}}
\newcommand{\Cscd}{\frac12 R_\mathsf{sc}}
\newcommand{\Rincd}{R_\mathsf{+}}
\newcommand{\Rfacd}{R_\mathsf{fac}}
\newcommand{\rnote}[1]{}

\usepackage[skip=12pt]{parskip}
\usepackage[letterpaper,top=2cm,bottom=2cm,left=3cm,right=3cm]{geometry}

\usepackage{pifont}
\DeclareFontFamily{U}{dice3d}{}
\DeclareFontShape{U}{dice3d}{m}{n}{<-> s*[4] dice3d}{}

\makeatletter
\algnewcommand\algorithmicwhen{\textbf{when}}\algdef{SE}[WHEN]{When}{EndWhen}[1]{\algpx@startIndent\algpx@startCodeCommand\algorithmicwhen\ #1\ \algorithmicdo}{\algpx@endIndent\algpx@startCodeCommand\algorithmicend\ \algorithmicwhen}

\ifbool{algpx@noEnd}{\algtext*{EndWhen}\algtext*{EndCase}\apptocmd{\EndWhen}{\algpx@endIndent}{}{}}{}

\pretocmd{\When}{\algpx@endCodeCommand}{}{}

\ifbool{algpx@noEnd}{\pretocmd{\EndWhen}{\algpx@endCodeCommand[1]}{}{}}{\pretocmd{\EndWhen}{\algpx@endCodeCommand[0]}{}{}}\makeatother
\usetikzlibrary{shapes.geometric, positioning}
\usepackage{pstricks}
\usetikzlibrary{intersections, calc, angles}
\makeatletter
\usetikzlibrary{decorations,backgrounds}
\def\pgfdecoratedcontourdistance{0pt}
\pgfset{
  decoration/contour distance/.code=\pgfmathsetlengthmacro\pgfdecoratedcontourdistance{#1}}
\pgfdeclaredecoration{contour lineto closed}{start}{\state{start}[
    next state=draw,
    width=0pt,
    persistent precomputation=\let\pgf@decorate@firstsegmentangle\pgfdecoratedangle]{\pgfpathmoveto{\pgfpointlineattime{.5}
      {\pgfqpoint{0pt}{\pgfdecoratedcontourdistance}}
      {\pgfqpoint{\pgfdecoratedinputsegmentlength}{\pgfdecoratedcontourdistance}}}}\state{draw}[next state=draw, width=\pgfdecoratedinputsegmentlength]{\ifpgf@decorate@is@closepath@ \pgfmathsetmacro\pgfdecoratedangletonextinputsegment{-\pgfdecoratedangle+\pgf@decorate@firstsegmentangle}\fi
    \pgfmathsetlengthmacro\pgf@decoration@contour@shorten{-\pgfdecoratedcontourdistance*cot(-\pgfdecoratedangletonextinputsegment/2+90)}\pgfpathlineto
      {\pgfpoint{\pgfdecoratedinputsegmentlength+\pgf@decoration@contour@shorten}
      {\pgfdecoratedcontourdistance}}\ifpgf@decorate@is@closepath@ \pgfpathclose
    \fi
  }\state{final}{}}
\makeatother
\tikzset{
  contour/.style={
    decoration={
      name=contour lineto closed,
      contour distance=#1
    },
    decorate}}

\usepackage{amsmath}
\usepackage{amsthm}
\usepackage{graphicx}
\newtheorem{claim}{Claim}
\newtheorem{theorem}{Theorem}

\newtheorem{corollary}{Corollary}
\usepackage[colorlinks=true, allcolors=blue]{hyperref}
\usepackage[capitalise]{cleveref}
\crefname{claim}{Claim}{Claims}
\title{To buy or not to buy: deterministic rent-or-buy problems on node-weighted graphs}
\author{Sander Borst\\MPI \and Moritz Venzin\\Bocconi}
\newcommand{\bd}{\mathsf{bd}}
\newcommand{\nmfl}{\mathsf{SC}}
\begin{document}
\maketitle

\begin{abstract}
  
    We study the \emph{rent-or-buy} variant of the online Steiner forest problem on \emph{node- and edge-weighted} graphs. For $n$-node graphs with at most $\bar{n}$ non-zero node-weights, and at most $\ktild$ different arriving terminal pairs, we obtain the following:
    \begin{itemize}
        \item A \emph{deterministic}, $O(\log n \log \bar{n})$-competitive algorithm. This improves on the previous best, $O(\log^4 n)$-competitive algorithm obtained by the black-box reduction from~\cite{BartalCharikarIndyk01} combined with the previously best deterministic algorithms for the simpler ``buy-only'' setting.
        \item A \emph{deterministic}, $O(\bar{n}\log \ktild)$-competitive algorithm. This generalizes the $O(\log \ktild)$-competitive algorithm for the purely edge-weighted setting from~\cite{umboh_online_2014}.
        \item A \emph{randomized}, $O(\log \ktild \log \bar{n})$-competitive algorithm. All previous approaches were based on the randomized, black-box reduction from~\cite{AwerbuchAzarBartal96} that achieves a $O(\log \ktild \log n)$-competitive ratio when combined with an algorithm for the ``buy-only'' setting.
    \end{itemize}
    Our key technical ingredient is a novel charging scheme to an instance of \emph{online prize-collecting set cover}. This allows us to extend the witness-technique of~\cite{umboh_online_2014} to the node-weighted setting and obtain refined guarantees with respect to $\bar{n}$, already in the much simpler ``buy-only'' setting.
\end{abstract}

\section{Introduction}
The \emph{Steiner forest} problem is a cornerstone of network design. 
Given a graph $G=(V,E)$ and pairs of terminals $T \subseteq V\times V$, the goal is to select a minimum weight subgraph that connects each pair. 
In the \emph{online, node-weighted} setting, there is a cost function $c: V \rightarrow \mathbb{R}_{\geq 0}$ on the vertices of the graph, and terminal pairs arrive one-by-one. After every arrival, we need to ensure that the pair is connected. 
Crucially, decisions cannot be altered in hindsight, i.e.\,we must \emph{augment} a set of selected vertices so that their induced subgraph connects all arrived pairs so far. The goal is to minimize the total cost of selected vertices. In this paper, we consider the online node-weighted Steiner forest problem in the \emph{rent-or-buy} setting. 
Given some parameter $M$, every vertex can either be \emph{rented} at cost $c_v$, in which case, we may only use it to connect the current pair of terminals, or \emph{bought} at cost $M\cdot c_v$, in which case it may be used indefinitely from then on. 
The goal is to minimize the total cost of rented vertices and bought vertices. More specifically, the aim is to minimize the \emph{competitive ratio}; the supremum of the ratio between the cost of the online algorithm, and the minimum cost of any feasible solution, for any sequence of terminal pairs.

This problem is centrally situated between several prominent problems in (online) network design.
Indeed, the node-weighted setting sits between the \emph{edge-weighted} setting and the $\emph{directed}$ setting\footnote{To reduce the edge-weighted to the node-weighted setting, we replace each edge by an unweighted edge subdivided by a node of node-weight equal to the original edge-weight. 
To reduce the node-weighted to the directed setting, we turn each edge into two (unweighted) directed edges, and split each vertex into $v_{\text{in}}$ and $v_{\text{out}}$. For each vertex $v$, we direct all incoming edges to $v_{\text{in}}$, all outgoing edges from $v_{\text{out}}$, and set the weight of the directed edge $v_{\text{in}}$ to $v_{\text{out}}$ to be $c_v$.\label{Footnote:edge_to_node}}. 
On the other hand, the rent-or-buy setting generalises over the simpler ``buy-only'' setting, and serves as a stepping stone for the even more general \emph{buy-at-bulk} setting.\footnote{The cost of selecting $x$ times an edge/vertex of cost $c$ is $f(x)\cdot c$, for some concave function $f(\cdot)$. 
The rent-or-buy setting is a special case by setting $f(x) := \min\{x, M\}$. The buy-only setting corresponds to $f(x) := \min\{x, 1\}.$} 
The latter is often also referred to as \emph{multicommodity buy-at-bulk}. 
There have been considerable efforts in devising online algorithms for these problems, both for the Steiner forest and the less general Steiner tree setting. 
We summarize those algorithmic results in Table~\ref{intro:state_of_the_art}.

\begin{table}[h!]
    \centering

            \renewcommand{\arraystretch}{1} \begin{tabular}{>{\centering\arraybackslash}m{1cm}!{\vrule width 2pt} *{3}{>{\centering\arraybackslash}p{3.5cm}}} & edge-weighted & node-weighted & directed \\ \hline
                Buy &  $\begin{matrix}
                    \\
                    \Theta(\log k)\\
                    \cite{ImaseWaxman91, BermanCoulston97}\\
                    \\
                \end{matrix}$  & $\begin{matrix}
                    \\
                    \Theta(\log k \log n)^{\raisebox{0.7mm}{\text{\faDiceTwo}\,,\,} \raisebox{0.7mm}{\text{\faEyeSlash}}} / \Theta(\log^2 n )\\
                    \cite{NaorPanigrahiSingh11, HajiaghayiLiaghatPanigrahi13, HajiaghayiLiaghatPanigrahi14, BorstEliasVenzin25}\\
                    \\
                \end{matrix}$   &   \\
                RoB & $\begin{matrix}
                    
                    \Theta(\log \ktild)\\
                    \cite{AwerbuchAzarBartal96, umboh_online_2014}
                    \\
                \end{matrix}$  &  $\begin{matrix}
                    \\
                    \Theta(\log^2 n)^{\raisebox{0.7mm}{\text{\faDiceSix}}} / O(\log^4 n)\\
                    \cite{BartalCharikarIndyk01, AwerbuchAzarBartal96} + \cite{bienkowski_nearly_2021, BorstEliasVenzin25}\\
                    \\
                \end{matrix}$&   \\
                BaB & $\begin{matrix}
                    \\
                    \Theta(\log \ktild)^{\raisebox{0.7mm}{\text{\faTree}}} \\
                    \cite{GuptaRaviTalwarUmboh_Spanners}\\
                    \\
                \end{matrix}$   &   & $\begin{matrix}
                    O(\text{poly}\log n)^{\raisebox{0.7mm}{\text{\faTree\,,\,\faDiceFour\,,\,\faEyeSlash\,,\,\faCalculator}}}\\
                    \cite{EneChakrabartyKrishnaswamyPanigrahi}
                \end{matrix}$
            \end{tabular}
    \caption{State of the art competitive ratios for online network design. The number of vertices in the graph is denoted by $n$, the number of arriving terminal pairs by $k$, and the number of \emph{distinct} terminal pairs by $\ktild$. Ratios written with a $\Theta(\cdot)$ are asymptotically tight. Results marked with~\faTree\ are only valid for the case of single-source (e.g. analog to Steiner tree). Results marked with a dice~\faDiceFive\ are randomized, the corresponding competitive ratio holds in expectation over the random bits of the algorithm. Results marked with~\faEyeSlash\ only hold against oblivious adversaries, e.g.\,the adversary cannot observe the algorithms actions and choose the next request adaptively. The algorithm for online directed buy-at-bulk only applies to single-source buy-at-bulk, the version for multicommodity being at least as hard as \textsc{Label Cover},~\cite{DodisKhanna}. It runs in \emph{quasi-polynomial} time~\faCalculator. }\label{intro:state_of_the_art}
\end{table}

The rent-or-buy setting is a common theme in decision making. Problems such as \emph{ski-rental}, \emph{TCP-acknowledgment}, \emph{snoopy caching} or the \emph{parking permit problem}, are prototypical examples, see for instance~\cite{KarlinKenyonRandallTCP, KarlinManasseMcGeoghOwicki90, Meyerson05}. More related to the context of network design is the \emph{file replication problem}, which corresponds to a rent-or-buy problem on a tree, see~\cite{BlackSleator89, BartalFiatRabani}. The first to study the rent-or-buy setting in the context of Steiner problems were Awerbuch, Azar and Bartal, see~\cite{AwerbuchAzarBartal96}. In fact, they show that for any online problem that can be cast as a type of \emph{metrical task system}, network design problems being a particular case, there is a simple procedure to adapt an online algorithm to the rent-or-buy setting. More specifically, given any online algorithm that is $c$-competitive against \emph{adaptive online adversaries} in the buy-only setting, their procedure yields a \emph{randomized} online algorithm for the rent-or-buy setting, that is $O(c)$-competitive against adaptive online adversaries.\footnote{In Section~\ref{subsec:adversaries} we give a formal overview of the adversarial settings.} Subsequently, this was further extended to yield a deterministic procedure that yields a $O(c^2)$-competitive deterministic online algorithm, see~\cite{BartalCharikarIndyk01}. 

Using these two procedures and the corresponding algorithms from the edge-weighted setting, see~\cite{ImaseWaxman91, BermanCoulston97}, this directly yields a $\Theta(\log \ktild)$-competitive randomized online algorithm for Steiner forest in the rent-or-buy setting, as well as an $O(\log^2 \ktild)$-competitive deterministic algorithm. Here, we denote by $\ktild$ the number of \emph{distinct} terminal pairs, as contrary to the buy-only setting, terminal pairs may reappear. Later, a $\Theta(\log \ktild)$-competitive deterministic online algorithm was given in~\cite{umboh_online_2014}.

Contrary to the edge-weighted setting, the competitive ratios of all known algorithms for online node-weighted Steiner problems (in the buy-only setting) depend on whether the adversary is adaptive or \emph{oblivious} (non-adaptive). 
The first approaches were inherently randomized and only hold against oblivious adversaries,~\cite{NaorPanigrahiSingh11, HajiaghayiLiaghatPanigrahi14}. The subsequent algorithm from~\cite{HajiaghayiLiaghatPanigrahi14}, was then made deterministic at a small loss, by giving a deterministic online algorithm for non-metric facility location, see~\cite{bienkowski_nearly_2021}. Consequently, since any deterministic online algorithm is competitive against adaptive adversaries, combined with the procedure from~\cite{AwerbuchAzarBartal96}, this yielded the first $O(\log^2 n \log \ktild)$-competitive algorithm for node-weighted Steiner forest in the rent-or-buy setting. 
The most recent online algorithm for node-weighted Steiner forest achieves a competitive ratio $\Theta(\log k \log n)$ against oblivious adversaries. It can be derandomized, resulting in a competitive ratio of $\Theta(\log^2 n)$, see~\cite{BorstEliasVenzin25}. We note that these two bounds are analogous to (and match) the $\Theta(\log k \log m)$- resp. $\Theta(\log n \log m)$-competitive ratios for online Set Cover in the randomized / deterministic setting\footnote{(Online) Set Cover reduces to (online) node-weighted Steiner forest / tree. 
$k$ corresponds to the number of arriving elements / pairs of arriving terminals, $m$ to the number of vertices in the graph / sets, and $n$ to the number of possible elements / possible pairs of terminals (\smash{$n = {m \choose 2}$}). See~\cite[Section~$2.2$]{BorstEliasVenzin25} for an illustration.},~\cite{alon_general_2006, Korman04, BuchbinderNaor09}. These bounds are asymptotically tight; in particular, it is not possible to obtain a deterministic algorithm with a competitive ratio of $O(\log k \log n)$. 

Applying the procedures from \cite{AwerbuchAzarBartal96, BartalCharikarIndyk01} to the deterministic algorithm from \cite{BorstEliasVenzin25} in a black-box way, we obtain a deterministic $O(\log^4 n)$ online algorithm, as well as a randomized, $O(\log^2 n)$-competitive online algorithm. It is also possible to do the same for the randomized version of this algorithm, yielding an $O(\log \ktild \log n)$-competitive algorithm. However, this needs a more careful argument. We will elaborate on this in \cref{subsec:related_work}.

This leaves quite a gap between the randomized and the deterministic setting. As suggested by the edge-weighted setting, it might indeed be possible to completely match the respective competitive ratios from the buy-only setting. In this paper we confirm this. In fact, we obtain a refined bound which depends on $\bar{n}$, the number of non-zero node-weights\footnote{This also applies to nodes of small degree: any node $v$ with constant degree can be replaced by a node with node-weight $0$, where we add $c_v$ to its incident edge-weights. This increases the overall cost by at most a constant.}. Our main result is a deterministic, $O(\log n \log\bar{n})$-competitive online algorithm for Steiner forest in the rent-or-buy setting. This even improves over the respective algorithm in~\cite{BorstEliasVenzin25} for the buy-only setting. It can be modified to give a deterministic $O(\bar{n}\log \ktild)$-competitive algorithm, which generalizes the deterministic algorithm of~\cite{umboh_online_2014} to the edge- and node-weighted setting. Finally, our last result is a randomized, $O(\log \ktild \log \bar{n})$-competitive algorithm against oblivious adversaries. 

\subsection{Our results}\label{subsec:our_results}
We now formally present our results. We consider edge- and node-weighted graphs with $n$ vertices, denote by $\bar{n}$ the number of nodes with non-zero node-weight, and denote by $\ktild$ the number of distinct terminals pairs arriving. 

All our results are obtained through the same algorithmic framework, but using three different algorithms for prize-collecting set cover as a subroutine. We present the three resulting results separately.

\begin{theorem}
\label{intro:main_thm_adaptive}
There is a deterministic, polynomial-time online algorithm for
node-weighted Steiner forest in the rent-or-buy setting
which is $O(\log n \log \bar{n})$-competitive
against an adaptive adversary. If the terminal pairs are known to come from some subset $T \subseteq V \times V$, the competitive ratio can be improved to $O(\log |T| \log \bar{n})$. 
\end{theorem}
We consider this to be our main result. All previously known deterministic algorithms use~\cite{BartalCharikarIndyk01}, and were a multiplicative factor of $O(\log^2 n)$ worse than the corresponding randomized version. This result closes this gap. We also note that it matches the currently best deterministic algorithm for the buy-only version (e.g. the special case $M = 1$) from~\cite{BorstEliasVenzin25}, and even slightly improves on it, as the dependency is now with respect to $\log\bar{n}$, instead of $\log n$.

The competitive ratio can be slightly improved to $O(\log\ktild\log\bar{n})$, to depend on the number of distinct \emph{actually} arriving terminals pairs $\ktild$, instead of the number of distinct terminal pairs $|T|$ that \emph{could} arrive. However, to do so, it is necessary to rely on randomization, as already demonstrated in the corresponding lower bounds for online set cover in~\cite{Alon_et_al_SetCover}. This again matches (and slightly improves) the corresponding result from~\cite{BorstEliasVenzin25} for the buy-only setting.

\begin{theorem}
\label{intro:main_thm_oblivious}
There is a randomized, polynomial-time online algorithm for
node-weighted Steiner forest in the rent-or-buy setting
which is $O(\log \ktild \log \bar{n})$-competitive
against an oblivious adversary. 
\end{theorem}
If one does not care about the logarithmic dependency on $\bar{n}$, it is possible to make this deterministic. The resulting competitive ratio is $O(\bar{n}\cdot \log\ktild)$, it depends \emph{linearly} on the number of non-zero node-weights. This recovers the deterministic, $O(\log\ktild)$-competitive algorithm of~\cite{umboh_online_2014} in the edge-weighted setting.

\begin{theorem}\label{intro:main_cor_det}
    There is a deterministic, polynomial-time online algorithm for node-weighted Steiner forest in the rent-or-buy setting which is $O(\bar{n} \cdot \log \ktild)$-competitive against an adaptive adversary.
\end{theorem}

\subsection{Our techniques}

On a high-level, we combine the \emph{witness} technique from the edge weighted variant for rent-or-buy from~\cite{umboh_online_2014} with an instance of \emph{prize-collecting set cover}. Before describing our approach, we briefly describe the approach for the edge-weighted setting, as well as the relevant techniques necessary for handling the node-weighted setting.

For the rent-or-buy setting, the difficulty lies in deciding whether it is worthwhile to buy an edge since we can reuse it at least $M$ times in the future, or whether we should rent. 
In the Ski Rental problem, e.g. rent-or-buy on a graph consisting of a single edge, a simple way to decide this and guarantee $2$-competitiveness is as follows: we rent $M$ times, then we buy. On a high level, this charging scheme was adapted to the edge-weighted Steiner problems in~\cite{umboh_online_2014}, surprisingly by simple greedy procedure (though the analysis is quite involved!): if the new request $r$ lies close to $M$ previous requests for which we rented, buy the greedy path the algorithm selects for~$r$. 
The cost of the bought path can be charged against these $M$ witnesses and all future requests close to any of them can be routed through the bought path. 
This ensures there cannot appear any clusters of $\Omega(M)$ rented paths for which it would have been beneficial to just buy once, allowing for a charging scheme against the optimal solution.

Such a greedy approach cannot work in the node-weighted setting, even in the buy-only setting. Because the weights are on the nodes, it is possible for some vertex $v$ to be close to many (eventually) arriving terminals, but, all these terminals are quite far from each other. 
Indeed, this happens if these terminals have small distance to~$v$, but $v$ has significant weight. Hence, even if we always buy a path for each arriving request, we might still end up with huge clusters of terminals for which an optimal offline solution only buys a single node. 
Hence, it is necessary to strategically buy certain nodes, which will decrease the distance of future requests. As shown in~\cite{NaorPanigrahiSingh11, HajiaghayiLiaghatPanigrahi13}, these nodes can be selected through an instance of online \emph{non-metric facility location}. 

We adapt these ideas as follows. We strategically buy nodes as in the buy-only setting, however, it is \emph{necessary} that we base our decision on an instance of online \emph{prize-collecting set cover}, a generalisation of the classical online set cover or non-metric facility location problem used in previous works,~\cite{NaorPanigrahiSingh11, HajiaghayiLiaghatPanigrahi14, BorstEliasVenzin25}. For each arriving terminal, such an instance can encode whether we should rent (pay the rental cost as penalty), or whether we should strategically buy nodes (at cost at least $M$ times the penalty). Crucially, such a prize-collecting variant of set cover can handle reappearances of terminals in case we opt to rent. In parallel, we run a charging scheme similar to that in~\cite{umboh_online_2014} to ensure that no more than $M$ terminals can ``cluster'', carefully balancing their cost against the cost of the prize-collecting instance of set cover. Finally, as a by-product of our conceptual simplification of using online (prize-collecting) set cover instead of (prize-collecting) non-metric facility location, our obtained bounds depend on $\bar{n}$, the number of non-zero node-weights, instead of $n$, the total number of vertices. This results in a framework that truly generalizes and recovers the best-known results from the edge-weighted setting, without incurring an extra $O(\log n)$ factor as in previous works.

\subsection{Related work}\label{subsec:related_work}

Most relevant to our work is the randomized procedure from~\cite{AwerbuchAzarBartal96}. This procedure can convert any algorithm for the buy-only setting into a randomized algorithm for the rent-or-buy setting with only a constant loss in the competitive ratio. It works in a very simple way: each request is passed to the buy-only algorithm with probability $1/M$, and for the remaining requests the algorithm rents. In their original paper, they only claim competitiveness against online adaptive adversaries when the algorithm for the buy-only setting is competitive against online adaptive adversaries. This would rule out using the $O(\log k\log n)$-competitive algorithm from~\cite{BorstEliasVenzin25}, as that algorithm is not competitive against adaptive adversaries. However, a slightly more careful analysis of the result from \cite{AwerbuchAzarBartal96} reveals that the same procedure can be applied, resulting in a randomized, $\Theta(\log \ktild \log n)$-competitive algorithm against oblivious adversaries. This is a weaker adversarial setting than that of an adaptive online adversary. We elaborate on this in Appendix~\ref{app:bbreduction}. \section{Preliminaries}

This section is divided in three parts. First, in Section~\ref{subsec:graph}, we present the necessary background from graph theory. In Section~\ref{subsec:online-problems}, we give a formal description of online problems under consideration. 
Finally, in Section~\ref{subsec:adversaries}, we describe the relevant adversarial settings

\subsection{Graph Theory}\label{subsec:graph}
We consider undirected graphs $G = (V, E)$ with a node-weight function $c: V \rightarrow \mathbb{R}_{\geq 0}$. This setting also captures edge-weights and is purely for convenience. For any pair of vertices~$u, v\in V$, we set $d_G(u,v)$ to be the distance with respect to $c(\cdot)$ between $u$ and $v$, excluding the weights of $u$ and $v$. Given some set $A \subset V$, the subgraph induced by $A$ is denoted by $G[A]$, and consists of all vertices in $A$ as well as all edges from $E$ that have both endpoints in $A$. $G/A$ denotes the graph $G$ where the weight of all nodes in $A$ have been set to $0$, and $d_{G/A}(u,v)$ denotes the shortest distance between $u$ and $v$ in $G/A$, again excluding the weights of $u$ and $v$. Finally, given some radius $r > 0$ and some vertex $u$, we define the open ball of radius $r$ around $u$ to be $B(u, r):=\{v \in V \mid d_{G}(u,v) + c_v \leq r\}$, and its boundary to be $\bd B(u, r):=\{v \in V \mid d_{G}(u,v) \leq r < d_G(u,v) + c_v\}$. We set $\bar{B}(u,r) := B(u, r) \cup \bd B(u, r)$.  

\subsection{Online problems}\label{subsec:online-problems}
In the online node-weighted Steiner forest problem ($\nwsf$) in the rent-or-buy setting, we are given a node-weighted graph upfront, as well as a parameter $M$. 
Then, one-by-one, we receive pairs of terminals $(s_1, t_1), \ldots, (s_k, t_k) \subseteq V\times V$. At each timestep $i\in [k]$, we need to maintain (augment) a set $B$ of bought nodes and select a set of nodes $R_i$ such that $(s_i, t_i)$ is connected in $G[B \cup R_i]$. The total cost is \smash{$M\cdot c(B) + \sum_{i \geq 1}^k c(R_i)$}. 
The \emph{Steiner tree} problem is defined analogous, except that we need to maintain a \emph{connected} subgraph connecting all the arrived terminals $s_1, \ldots, s_k$; this corresponds to the Steiner forest problem in the setting where the terminal pairs are of the form $(s_1, s_2), (s_1, s_3), \ldots, (s_1, s_k)$.

An instance of the set cover problem ($\SC$) is given by a set system $(X, \S, \cost)$, where $X$ is a groundset of elements, $\mathcal{S}\subseteq 2^X$ is a set of subsets of $X$, and $\cost: \mathcal{S} \rightarrow \mathbb{R}_{\geq 0}$ is a cost function. In the online setting, the elements are revealed one-by-one, and we need to ensure it is covered by a set. Sets from the cover cannot be removed, and the goal is to minimize the total cost of the selected sets. This problem has been studied in two settings. In the first, we have \emph{no} knowledge of the set system. Only upon arrival of an element, we learn which sets contain it. In the second setting, we learn the set system $(X, \S)$ upfront. The subset of elements in $X$ that needs to be covered is revealed in an online fashion. Deterministic algorithms for online set cover are only meaningful in the second setting. In contrast, randomized algorithms against oblivious adversaries also apply in the first setting.

\subsection{Adversarial settings}\label{subsec:adversaries}

An online problem is a game between an online algorithm and an adversary. 
The adversary reveals the input $I_\adv$ to some problem piece by piece, after each reveal, the online algorithm must augment their current solution to maintain feasibility. The cost incurred by the online algorithm is compared to that of the adversary, its ratio is referred to as the \emph{competitive ratio}. This notion does depend on the type of adversary, i.e. how the adversary is allowed to generate the sequence $I_\adv$ and how it has to service it. We give a brief description of three adversarial settings that we consider in this paper, as well as the corresponding notion of competitiveness. For a more detailed exposition, see~\cite{BorodinYaniv98}.

\paragraph*{Adaptive adversaries} An adaptive adversary can create the input sequence $I_\adv$ adaptively. In each time-step, it selects the next request to be presented to the online algorithm based on previous actions of the online algorithm so far. 

\paragraph*{Semi-adaptive adversaries} This adversarial setting was introduced recently in~\cite{BorstEliasVenzin25}. 
A semi-adaptive adversary commits to a super-instance $\hat{I}_\adv$ of the input sequence. In the online phase, it adaptively select $I_\adv \subseteq \hat{I}_\adv$ to be presented to the online algorithm. This adversarial setting is relevant in settings where the competitive ratio depends on the size of $\hat{I}_\adv$, such as online set cover against an adaptive adversary.

\paragraph*{Oblivious adversaries} An oblivious adversary has to commit to the whole input sequence upfront, and cannot change it based on the actions of the online algorithm.

To define the competitive ratio, we further distinguish between \emph{online} and \emph{offline} adversaries. An online adversary needs to service their own request sequence $I_\adv$ in an online fashion, i.e.\,cannot wait to observe the actions of the online algorithm and service their requests accordingly, while an offline adversary can service their request sequence $I_\adv$ after it is finished presenting it to the online algorithm. Denoting by $\mathsf{alg}(I_\adv)$ the cost of the online algorithm and by $c(\adv)$ the cost of the (respective) adversary, the online algorithm is said to be $\gamma$-competitive if
\begin{equation*}
    \mathbb{E}_{\mathsf{alg}}[\mathsf{alg}(I_\adv)] \leq \gamma\cdot\mathbb{E}_{\mathsf{alg}}[c(\adv)].
\end{equation*}
Note that the expectation of the right-hand-side is also over the randomness of the actions, since, whenever the adversary is adaptive, the actions of the online algorithm may influence future requests. If the adversary is oblivious, the right-hand-side holds deterministically, and if the online algorithm is deterministic, both sides hold deterministically. \section{Online prize-collecting set cover}
In this paper, we will make use of the \emph{prize-collecting} variant of the online set cover problem. 

This is essentially the online set cover problem, with a small twist: for each arriving element $e$, there is a penalty $p$ if $e$ is left uncovered. Crucially, elements can reappear (with varying penalties). Formally, such an instance is given by \smash{$(\tilde{X}, \S, \cost: \S \rightarrow \mathbb{R}_{\geq 0}, p: \tilde{X} \times \mathbb{N} \rightarrow \mathbb{R}_{\geq 0})$}. Here, we denote the ground set by $\tilde{X}$ to emphasize the difference to the (multi-)set of eventually arriving elements. We will refer to $\ktild$ as the number of distinct arriving elements.

It is important to note that the naive reduction from (online) set cover, i.e.\,duplicate each element proportional to its number of reappearances and add singleton sets with corresponding penalty cost, does not yield any meaningful competitive ratio. Indeed, such a reduction creates a number of elements which is proportional to the number of arriving elements, which may be considerably higher than $\ktild$. As an example, for the $\nwsf$ problem in the rent-our-buy setting with parameter $M$, this can be of order $M\cdot \tilde{k}$, which would result in a competitive ratio that depends on $M$. 

\subsection{Background on online set cover}\label{subsec:backgroundNMFL}

All approaches for online set cover are based on the following general framework. Throughout the online phase, a monotonically increasing fractional solution is maintained, that is feasible for all arrived elements. In parallel, a \emph{rounding} scheme is run to convert these fractional values to an actual integral solution, i.e.\,actually picking sets covering the arrived elements. For the first part, we rely on the following result.

\begin{theorem}[Theorem 3.2,~\cite{buchbinder_online_2009}]\label{algo:frac_sc}
    There is a deterministic online algorithm that maintains a monotonically increasing fractional solution to the online set cover problem. The algorithm is $O(\log\ell)$-competitive with respect to an optimal fractional solution, where~$\ell := \max_{x\in X}|\{S\in\S | x \in S\}|$ is the maximum number of sets that an element is contained in.
\end{theorem}

Note that this fractional solution is compared to the best fractional solution in hindsight, and can be maintained deterministically. Hence this holds against an offline adversary in any adversarial model. In contrast, the rounding part does depend on the adversarial setting. Against an oblivious or semi-adaptive adversary, the rounding scheme is randomized and the guarantees hold in expectation. However, contrary to the deterministic rounding scheme (against adaptive adversaries), the underlying set system does not need to be revealed upfront. We resume these results in the following two theorems.

\begin{theorem}[Theorem $3$,~\cite{BorstEliasVenzin25}]\label{algo:nmfl_semiadaptive}
    There is a randomized procedure that rounds a monotonically increasing solution to an instance of online set cover. Against a semi-adaptive adversary that commits to a super-instance $\hat{I} = (\hat{X},\mathcal{S}, \cost)$ upfront and adaptively selects $X' \subseteq \hat{X}$ to be presented to the algorithm, we have that:
    \begin{alignat*}{1}
        \mathbb{E}[\cost(\mathsf{alg})] \leq O(\log|\hat{X}|) \cdot \mathbb{E}[\mathsf{val}_{\mathsf{frac}}].
    \end{alignat*}
    Here, $\mathsf{val}_{\mathsf{frac}}$ is the cost of the fractional solution provided to the procedure.
\end{theorem}
Note that this generalises the result for online set cover against oblivious adversaries. If the adversary is oblivious to the actions of the rounding procedure, their actions may as well be fixed. Hence, we may assume that $X' = \hat{X}$ and that the right-hand-side holds deterministically. In that case, combined with Theorem~\ref{algo:frac_sc}, this recovers the $O(\log k \log \ell)$-competitive ratio against (offline) oblivious adversaries of~\cite{BuchbinderNaor09}, where $k = |X'|$ is the number of arriving elements. Against adaptive adversaries, we rely on the following.

\begin{theorem}[Theorem 5.2, \cite{BuchbinderNaor09}]\label{algo:det_sc}
    There is a deterministic online rounding procedure that rounds a monotonically increasing solution to an instance of online set cover, provided the set system $\hat{I} = (\hat{X},\mathcal{S}, \cost)$ is revealed upfront. The adversary adaptively selects a subset of the elements that will arrive. The algorithm satisfies:
    \begin{alignat*}{1}
        \mathbb{E}[\cost(\mathsf{alg})] \leq O(\log|\hat{X}|) \cdot \mathbb{E}[\mathsf{val}_{\mathsf{frac}}].
    \end{alignat*}
    Here, $\mathsf{val}_{\mathsf{frac}}$ is the cost of the fractional solution provided to the procedure.
\end{theorem}

\subsection{Algorithm for prize-collecting online set cover}
We now give an online algorithm for the prize-collecting variant of online set cover.

We begin by describing an auxiliary instance of set cover $(X, \hat{\S})$. For each reappearance of an element $e\in \tilde{X}$ with penalty $p$, there is an element $(e, p) \in X$. For each set $S \in \S$, there is a set $\hat{S}\in \hat{\S}$ of same weight, so that $\hat{S}$ contains all copies of the elements contained in $S$, i.e.\,$(e,p)\in\hat{S} \implies e\in S$. Finally, for each element corresponding to $(e,p) \in \tilde{X}$, there is a singleton set $\hat{S}_{(e,p)}$ of weight $p$ in $\hat{\S}$. This construction is illustrated in Figure~\ref{fig:aux_instance_sc}. Clearly, any integral (fractional) solution to this auxiliary instance corresponds to an integral (fractional) solution to the original instance of prize-collecting set cover of same cost and vice-versa. 

\begin{figure}[htbp]
    \centering
    \begin{subfigure}[b]{0.45\textwidth}
\begin{tikzpicture}

\def\radius{.45cm}
\node[label = {[below, yshift = -0.7cm]{$e_1$}}, rectangle, inner sep = 2pt, fill] (e1) at (-1.5, 0) {};
\coordinate[label = {[below, yshift = -0.7cm]{$e_2$}}, rectangle, inner sep = 2pt, fill] (e2) at (-0.5, 0) {};
\coordinate[label = {[below, yshift = -0.7cm]{$e_3$}}, rectangle, inner sep = 2pt, fill] (e3) at (0.5, 0) {};
\coordinate[label = {[below, yshift = -0.7cm]{$e_4$}}, rectangle, inner sep = 2pt, fill] (e4) at (1.5, 0) {};
\coordinate[] (s1) at (-0.5, 2) {};
\coordinate[] (s2) at (-1.5, 2) {};
\coordinate[] (s3) at (1.5, 2) {};
\coordinate[] (s4) at (1.5, 2) {};

\coordinate (e1e2) at ($(e1)!\radius!-90:(e2)$);
\coordinate (e2e1) at ($(e2)!\radius!90:(e1)$);
\draw (e1e2) -- (e2e1);
\coordinate (e2s1) at ($(e2)!\radius!-90:(s1)$);
\coordinate (s1e2) at ($(s1)!\radius!90:(e2)$);
\draw (e2s1) -- (s1e2);
\coordinate (s1e1) at ($(s1)!\radius!-90:(e1)$);
\coordinate (e1s1) at ($(e1)!\radius!90:(s1)$);
\draw (s1e1) -- (e1s1);
\pic [draw, angle radius=\radius] {angle=e2e1--e2--e2s1};
\pic [draw, angle radius=\radius] {angle=s1e2--s1--s1e1};
\pic [draw, angle radius=\radius] {angle=e1s1--e1--e1e2};

\coordinate (e4s3) at ($(e4)!\radius!-90:(s3)$);
\coordinate (s3e4) at ($(s3)!\radius!90:(e4)$);
\draw (e4s3) -- (s3e4);
\coordinate (s3s4) at ($(s3)!\radius!-90:(s4)$);
\coordinate (s4s3) at ($(s4)!\radius!90:(s3)$);
\draw (s3s4) -- (s4s3);
\coordinate (s4e4) at ($(s4)!\radius!-90:(e4)$);
\coordinate (e4s4) at ($(e4)!\radius!90:(s4)$);
\draw (s4e4) -- (e4s4);
\pic [draw, angle radius=\radius] {angle=s3e4--s3--s3s4};
\pic [draw, angle radius=\radius] {angle=s4s3--s4--s4e4};
\pic [draw, angle radius=\radius] {angle=e4s4--e4--e4s3};

\def\radius{.35cm}
\coordinate (e2e3) at ($(e2)!\radius!-90:(e3)$);
\coordinate (e3e2) at ($(e3)!\radius!90:(e2)$);
\draw (e2e3) -- (e3e2);
\coordinate (e3e4) at ($(e3)!\radius!-90:(e4)$);
\coordinate (e4e3) at ($(e4)!\radius!90:(e3)$);
\draw (e3e4) -- (e4e3);
\coordinate (e4s2) at ($(e4)!\radius!-90:(s2)$);
\coordinate (s2e4) at ($(s2)!\radius!90:(e4)$);
\draw (e4s2) -- (s2e4);
\coordinate (s2e2) at ($(s2)!\radius!-90:(e2)$);
\coordinate (e2s2) at ($(e2)!\radius!90:(s2)$);
\draw (s2e2) -- (e2s2);
\pic [draw, angle radius=\radius] {angle=e3e2--e3--e3e4};
\pic [draw, angle radius=\radius] {angle=e4e3--e4--e4s2};
\pic [draw, angle radius=\radius] {angle=s2e4--s2--s2e2};
\pic [draw, angle radius=\radius] {angle=e2s2--e2--e2e3};

\end{tikzpicture}
\subcaption{The prize-collecting set cover instance $(\tilde{X}, \S)$}
\end{subfigure}
\hfill
\begin{subfigure}[b]{0.45\textwidth}
        \centering
\begin{tikzpicture}
\coordinate[] (s1) at (-0.5, 2) {};
\coordinate[] (s2) at (-1.5, 2) {};
\coordinate[] (s3) at (1.5, 2) {};
\coordinate[] (s4) at (1.5, 2) {};
 
\def\radius{.45cm}
\node[label = {[below, yshift = -0.7cm]{$e_{1,1}$}}, rectangle, inner sep = 2pt, fill] (e1) at (-2.5, 0) {};
\node[label = {[below, yshift = -0.7cm]{$e_{2,1}$}}, rectangle, inner sep = 2pt, fill] (e2) at (-1.5, 0) {};
\coordinate[label = {[below, yshift = -0.7cm]{$e_{2,2}$}}, rectangle, inner sep = 2pt, fill] (e3) at (-0.5, 0) {};
\coordinate[label = {[below, yshift = -0.7cm]{$e_{3,1}$}}, rectangle, inner sep = 2pt, fill] (e4) at (0.5, 0) {};
\coordinate[label = {[below, yshift = -0.7cm]{$e_{4,1}$}}, rectangle, inner sep = 2pt, fill] (e5) at (1.5, 0) {};
\coordinate[label = {[below, yshift = -0.7cm]{$e_{4,2}$}}, rectangle, inner sep = 2pt, fill] (e6) at (2.5, 0) {};
\coordinate[label = {[below, yshift = -0.7cm]{$e_{4,3}$}}, rectangle, inner sep = 2pt, fill] (e7) at (3.5, 0) {};
\coordinate[] (s1) at (-0.5, 2) {};
\coordinate[] (s2) at (-2.5, 2) {};
\coordinate[] (s3) at (2.5, 2) {};
\coordinate[] (s4) at (1.5, 2) {};

\coordinate (e1e2) at ($(e1)!\radius!-90:(e2)$);
\coordinate (e2e1) at ($(e2)!\radius!90:(e1)$);
\draw (e1e2) -- (e2e1);
\coordinate (e2e3) at ($(e2)!\radius!-90:(e3)$);
\coordinate (e3e2) at ($(e3)!\radius!90:(e2)$);
\draw (e2e3) -- (e3e2);
\coordinate (e3s1) at ($(e3)!\radius!-90:(s1)$);
\coordinate (s1e3) at ($(s1)!\radius!90:(e3)$);
\draw (e3s1) -- (s1e3);
\coordinate (s1e1) at ($(s1)!\radius!-90:(e1)$);
\coordinate (e1s1) at ($(e1)!\radius!90:(s1)$);
\draw (s1e1) -- (e1s1);
\pic [draw, angle radius=\radius] {angle=e2e1--e2--e2e3};
\pic [draw, angle radius=\radius] {angle=e3e2--e3--e3s1};
\pic [draw, angle radius=\radius] {angle=s1e3--s1--s1e1};
\pic [draw, angle radius=\radius] {angle=e1s1--e1--e1e2};

\def\radius{.35cm}
\coordinate (e2e3) at ($(e2)!\radius!-90:(e3)$);
\coordinate (e3e2) at ($(e3)!\radius!90:(e2)$);
\draw (e2e3) -- (e3e2);
\coordinate (e3e4) at ($(e3)!\radius!-90:(e4)$);
\coordinate (e4e3) at ($(e4)!\radius!90:(e3)$);
\draw (e3e4) -- (e4e3);
\coordinate (e4e5) at ($(e4)!\radius!-90:(e5)$);
\coordinate (e5e4) at ($(e5)!\radius!90:(e4)$);
\draw (e4e5) -- (e5e4);
\coordinate (e5e6) at ($(e5)!\radius!-90:(e6)$);
\coordinate (e6e5) at ($(e6)!\radius!90:(e5)$);
\draw (e5e6) -- (e6e5);
\coordinate (e6e7) at ($(e6)!\radius!-90:(e7)$);
\coordinate (e7e6) at ($(e7)!\radius!90:(e6)$);
\draw (e6e7) -- (e7e6);
\coordinate (e7s2) at ($(e7)!\radius!-90:(s2)$);
\coordinate (s2e7) at ($(s2)!\radius!90:(e7)$);
\draw (e7s2) -- (s2e7);
\coordinate (s2e2) at ($(s2)!\radius!-90:(e2)$);
\coordinate (e2s2) at ($(e2)!\radius!90:(s2)$);
\draw (s2e2) -- (e2s2);
\pic [draw, angle radius=\radius] {angle=e3e2--e3--e3e4};
\pic [draw, angle radius=\radius] {angle=e4e3--e4--e4e5};
\pic [draw, angle radius=\radius] {angle=e5e4--e5--e5e6};
\pic [draw, angle radius=\radius] {angle=e6e5--e6--e6e7};
\pic [draw, angle radius=\radius] {angle=e7e6--e7--e7s2};
\pic [draw, angle radius=\radius] {angle=s2e7--s2--s2e2};
\pic [draw, angle radius=\radius] {angle=e2s2--e2--e2e3};

\def\radius{.45cm}
\coordinate (e5e6) at ($(e5)!\radius!-90:(e6)$);
\coordinate (e6e5) at ($(e6)!\radius!90:(e5)$);
\draw (e5e6) -- (e6e5);
\coordinate (e6e7) at ($(e6)!\radius!-90:(e7)$);
\coordinate (e7e6) at ($(e7)!\radius!90:(e6)$);
\draw (e6e7) -- (e7e6);
\coordinate (e7s3) at ($(e7)!\radius!-90:(s3)$);
\coordinate (s3e7) at ($(s3)!\radius!90:(e7)$);
\draw (e7s3) -- (s3e7);
\coordinate (s3e5) at ($(s3)!\radius!-90:(e5)$);
\coordinate (e5s3) at ($(e5)!\radius!90:(s3)$);
\draw (s3e5) -- (e5s3);
\pic [draw, angle radius=\radius] {angle=e6e5--e6--e6e7};
\pic [draw, angle radius=\radius] {angle=e7e6--e7--e7s3};
\pic [draw, angle radius=\radius] {angle=s3e7--s3--s3e5};
\pic [draw, angle radius=\radius] {angle=e5s3--e5--e5e6};

\draw (e1) circle (0.25cm);
\draw (e2) circle (0.25cm);
\draw (e3) circle (0.25cm);
\draw (e4) circle (0.25cm);
\draw (e5) circle (0.25cm);
\draw (e6) circle (0.25cm);
\draw (e7) circle (0.25cm);

\end{tikzpicture}

\subcaption{The auxiliary set cover instance $(X, \hat{\S})$}
\end{subfigure}
\caption{In the prize-collecting instance, element $e_1$, appears once, $e_2$ twice, $e_3$ once, and $e_4$ three times. In the auxiliary instance, each copy is contained in the same sets as the original instance, but is also contained in a singleton set of weight corresponding to its penalty.}
\label{fig:aux_instance_sc}
\end{figure}
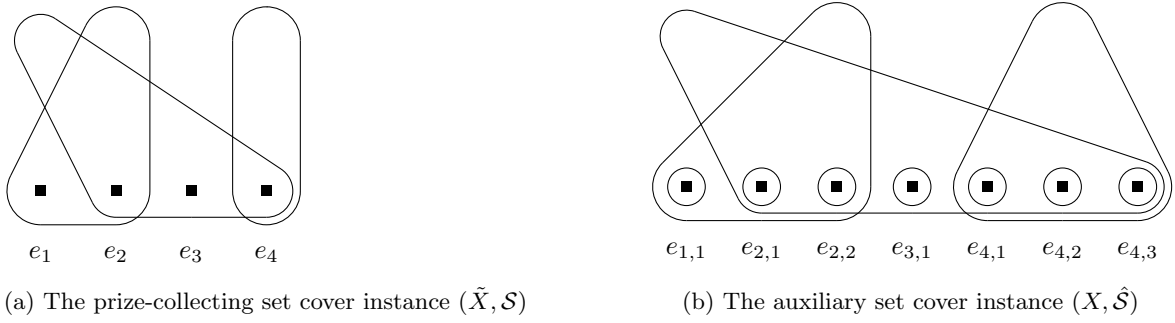

We now describe our main algorithm for online prize-collecting set cover on instance $(\tilde{X}, \S)$. We start by initializing an (empty) auxiliary instance $I_{\text{aux}}$ of online fractional set cover, $(X, \hat{\S})$, for which we will maintain a monotonically increasing fractional solution~\smash{$\{x_{\hat{S}}\}_{\hat{S}\in\hat{\S}}$} by the algorithm of~\cite{BuchbinderNaor09}, see Theorem~\ref{algo:frac_sc}. We also initialize an instance $\tilde{I}$ of online set cover on the same set system as $(\tilde{X}, \S)$. We will maintain monotonically increasing values $\{x_S\}_{S\in\S}$, for which we will use a rounding scheme for online set cover such as~\cite{Alon_et_al_SetCover, BuchbinderNaor09, BorstEliasVenzin25} to decide which sets to select for $(\tilde{X}, \S)$. Specifically, upon the arrival of an element $(e, p) \in \tilde{X}$, we pass a copy of the element to the auxiliary instance of set cover $(X, \hat{\S})$ as detailed above, and we update our fractional solution. If the fractional value of the singleton set containing $(e,p)$ is at least $1/2$, i.e.
\begin{equation*}
    x_{\hat{S}_{(e,p)}} > \tfrac{1}{2},
\end{equation*}
we pay the penalty. In parallel, we set $x_S = 2\cdot x_{\tilde{S}}$, i.e.\,the fractional value of each set $S\in\S$ is twice that of its corresponding set in $\hat{\S}$. We add any sets that the rounding scheme on $(\tilde{X}, \S)$ selects based on $\{x_S\}_{S\in \S}$. 

\begin{theorem}\label{thm:rand_prize_collecting}
    There is a randomized algorithm for online prize-collecting set cover against a semi-adaptive adversary that commits to super-instance $(\tilde{X}, \S, \cost)$, but not on the values of the respective penalties. We have that
    $$\mathbb{E}[c(\mathsf{ALG})] \leq O(\log |\tilde{X}| \log |\S|)\cdot \mathbb{E}[c(\opt_{\mathsf{frac}})],\,\text{and},$$
    $$c(\mathsf{ALG}_{\mathsf{pen}}) \leq O(\log |\S|) \cdot \mathbb{E}[c(\opt_\mathsf{frac})].$$
    Here $c(\mathsf{ALG})$ denotes the total cost incurred by the algorithm, $c(\mathsf{ALG}_{\mathsf{pen}})$ the cost incurred on penalties, and $\opt_{\mathsf{frac}}$ an optimum fractional solution to the sub-instance presented to the algorithm. If the set system $(\tilde{X}, \S, \cost)$ is revealed upfront, the algorithm can be made deterministic and the above competitive ratio holds deterministically.
\end{theorem}

\begin{proof}
    By Theorem~\ref{algo:frac_sc}, the fractional solution $\{x_{\hat{S}}\}_{\hat{S}\in\hat{\S}}$ for the auxiliary instance $I_{\mathsf{aux}}$ is $O(\log \ell)$ competitive, where $\ell$ is the maximum number of sets any element is contained in. Note that this is deterministic and does not require the assumption that the instance be revealed upfront. By construction, $\ell$ is at most $|\S| + 1$. Since we only pay the penalty if the singleton set has weight at least $1/2$, the total cost incurred on penalties is at most $$O(\log |\S|)\cdot c(\opt_\text{frac}).$$ Hence, whenever we do not pay the penalty for element $(e,p)$, the fractional weight of all sets $\hat{S} \ni (e,p)$ except for the singleton set must sum up to at least $1/2$. Since $x_S = 2\cdot x_{\hat{S}}$, the values $\{x_S\}_{S \in \S}$ fractionally cover $e$. Their fractional cost is $O(\log|\S|)\cdot\opt_{\text{frac}}$. Hence, by Theorem~\ref{algo:nmfl_semiadaptive} resp.~\ref{algo:det_sc}, these values are rounded in an online fashion (ensuring $e$ is covered), incurring an $O(\log |\tilde{X}|)$ multiplicative loss. This concludes the proof.
\end{proof}

To conclude this section, we describe a deterministic, $O(|\S|)$-competitive algorithm for online prize-collecting set cover that is $1$-competitive with respect to penalties. Specifically, we use a \emph{dual} charging scheme to decide when to pay the penalty and when to buy which sets. The dual is given by (D), the primal, e.g.\,the fractional relaxation to $(X, \hat{S})$, by (P).

\begin{minipage}{0.45\textwidth}
\begin{gather}
    \min \quad   \sum_{S \in \hat{\S}} c_{\hat{S}}\cdot x_{\hat{S}} \notag\\  
\begin{align*}
    \rlap{\hspace{-1.2cm} (P)}\text{s.t.} \quad  \sum_{\hat{S}: (e,p) \in \hat{S}} x_{\hat{S}} &\geq 1, \quad \forall (e,p) \in X \\
     x_{\hat{S}} &\geq 0, \quad \forall \hat{S} \in \hat{\S}
\end{align*}
\end{gather}
\end{minipage}
\hfill
\begin{minipage}{0.45\textwidth}
\begin{gather}
    \max \quad \sum_{(e,p) \in X} y_{(e,p)} \notag\\
\begin{align*}
    \rlap{\hspace{-1.2cm} (D)} \text{s.t.} \quad  \sum_{(e,p) \in \hat{S}} y_{(e,p)} &\leq c_{\hat{S}}, \quad \forall \hat{S} \in \hat{\S} \\
     y_{(e,p)} &\geq 0, \quad \forall (e,p) \in X
\end{align*}
\end{gather}
\end{minipage}
\vspace{5mm}

The algorithm is very simple. Whenever an element $e\in \tilde{X}$ with penalty $p$ arrives, we pass $(e,p)$ to the fractional auxiliary instance of set cover described above. We raise the value of the corresponding dual variable $y_{(e,p)}$ until one of the dual constraints (indexed by $\hat{S} \in \hat{\S})$ becomes tight, i.e. 
\begin{equation*}
    \sum_{(e,p) \in \hat{S}} y_{(e,p)} = c_{\hat{S}}.
\end{equation*}
We buy the set in $\S$ corresponding to such a set $\hat{S}$, or pay the penalty if it is the singleton set $\hat{S}_{(e,p)}$. 

\begin{theorem}\label{thm:det_prize_collecting}
    If the set system $(\tilde{X}, \S, \cost)$ is not revealed upfront, there is a deterministic, $O(|\S|)$-competitive algorithm for prize collecting set cover. Furthermore, it holds that
    $$c(\mathsf{ALG}_{\mathsf{pen}}) \leq c(\opt_\mathsf{frac}),$$
    where $c(\mathsf{ALG}_{\mathsf{pen}})$ denotes the cost incurred on penalties, and $c(\opt_\mathsf{frac})$ denotes the optimal fractional cost to sub-instance presented to the online algorithm.
\end{theorem}

\begin{proof}
    Clearly, throughout the algorithm, we maintain a feasible dual solution $y_{(e,p)}$ for (D), and each arriving element is covered. To bound the cost, we observe that by weak duality, $(\text{D}) \leq (\text{P})$. Since we only pay the penalty $p$ for element $e$ if \smash{$y_{(e,p)} = p$}, the total amount spent on penalties (i.e.\,the singleton sets in $\hat{\S}$) is at most (D). To bound the total cost, we observe that each $y_{(e,p)}$ contributes its value to at most $|\S| + 1$ sets, meaning that $c(\mathsf{ALG}) \leq (|\S| + 1)\cdot (\text{D}) \leq O(|S|)\cdot \opt_{\mathsf{frac}}.$ 
\end{proof}
 \section{Rent-or-Buy for node-weighted Steiner forest}

In this part we show Theorems~\ref{intro:main_thm_adaptive},~\ref{intro:main_thm_oblivious} and~\ref{intro:main_cor_det}. The algorithm is in Section~\ref{subsec:algo}, the analysis in Section~\ref{subsec:analysis}. 

\subsection{The algorithm}\label{subsec:algo}

Our algorithm relies on an auxiliary instance of online prize-collecting set cover. We describe it in Algorithm~\ref{alg:nwrobsf}. 

\paragraph*{Set-up} We assume that for all arriving terminal pairs, we have that $1 \leq d_{G/A}(s_i, t_i) \leq \ktild^6$,
where $\ktild$ is the number of arriving terminal pairs (set $\ktild := |T|$ for Theorem~\ref{intro:main_thm_adaptive}), and $A$ is the set of bought nodes. This is without loss of generality, see Appendix~\ref{app:distreduction}. This allows us to group pairs $(s_i, t_i)$ into \emph{layers} $j \in L := \{0, 1, \ldots, 6\log\ktild\}$ corresponding to $\lfloor \log d_{G/A}(s_i,t_i) +1\rfloor$. For ease of presentation we transform the node- and edge-weighted graph into a purely node-weighted graph as follows. Each edge with non-zero edge-weight $c_e$ is subdivided by $\lceil 8\cdot 1/c_e + 1\rceil$ nodes, each of node-weight strictly less than $1/8$ and summing exactly to $c_e$. Note that this transformation is only done for the purpose of the description and analysis of the online algorithm, but is not necessary to construct explicitly. 

\paragraph*{An instance of online prize-collecting set cover ($\pcsc$)} We describe the set system $(\tilde{X}, \S)$. For each node $u\in G$, and each $j \in L := \{0, 1, \ldots, 6\log \ktild\}$, there is an element $r_{u,j}$. Also, for each node $v \in G$, there is a set $S_v$ with cost $\log(\ktild)\cdot M\cdot c_v$. We define $R_{u,j}$ to be the family of sets $S_v$ with $v\in \Boundary(u,\Rsc)$ for which $c_v\geq \Csc$. Element $r_{u,j}$ is contained in all sets in $R_{u,j}$ and its penalty is~$2^j$. Whenever clear, we drop the prize-collecting and simply refer to this instance as an instance of set cover ($\SC$).

\paragraph*{Description of the algorithm} We now give a short overview. 
We start by initializing an instance of prize-collecting set cover as described above, $|L|$ witness sets $F_0, F_1, \ldots, F_{|L|}$, as well as a counter $y_{w,j} \leftarrow 0$ for each node $w\in G$. 
Every time a terminal pair $(s, t)$ arrives, we proceed as follows. We compute $d := d_{G/A}(s,t)$, the distance of $s$ to $t$, and we refer to $j := \lfloor \log d +1\rfloor$ as the \emph{layer} the demand pair is on. 
We always rent the greedy path from $s$ to $t$ (hence satisfying the connectivity requirement), however, we will also buy some nodes. 
To decide this, we distinguish between two cases, depending on whether $s$ or $t$ is \emph{uncovered} on layer $j$:
$s$ is said to be uncovered on layer $j$, if no witness set on layer $j$ containing $s$ has been selected, and, for all vertices $v$ in the open ball of radius $\Rinc$ centered at $s$, the counter $y_{v,j}$ is strictly less than $M$; e.g. 
\begin{equation*}
    R_{s,j}\cap F_j = \emptyset, \text{ and }, \forall w \in B(s, \Rinc): y_{w,j} < M.
\end{equation*}
If this holds, we pass $s$ (or $t$, whatever is uncovered) to the auxiliary instance of set cover.

If the instance of set cover adds a new set to the cover, we buy the corresponding node and add it to  $F_j$. For all nodes $w\in B(s, \Rinc)$ we increase the counter $y_{w,j}$ by one. 
If this does not hold, i.e.\,both $s$ and $t$ are covered, we proceed as follows. We pick two nodes $w_s, w_t \in F_j$ on layer $j$, where $w_s$ (resp. $w_t$) lies in the closed ball of radius $\Rfac$ around $s$ (resp. $t$). If this set is empty, we add $s$ (resp. $t$) to the witness set $F_j$ and set $w_s = s$ (resp. $w_t = t$). We then buy the shortest path between $w_s$ and $w_t$ in $G/A$. 

\paragraph*{Size of set system} From the description of the algorithm, it is clear that the only elements $r_{u,j}$ that are released to the instance of $\pcsc$, correspond to some terminal $u$ in the original graph. Hence, the number of arriving elements of $\tilde{X}$ is at most $O(\ktild\cdot \log \ktild)$, and the number of elements that could potentially arrive is at most $O(|T|\cdot \log |T|)$. To see that we can bound the number of sets by $\bar{n}$, note that the nodes corresponding to non-zero edge-weights in the original graph have node-weights strictly less than $1/8$ and do not contain any elements. As such, they are irrelevant and can be disregarded. All remaining sets correspond to one of the $\bar{n}$ original nodes with non-zero node-weight. 

\begin{algorithm}
  \caption{Node-weighted Rent-Or-Buy Steiner Forest \label{alg:nwrobsf}}
\begin{algorithmic}[1]
  \State $F_j\gets \emptyset$ for all $j\in L$
  \State Initialize $\pcsc$, an algorithm for online prize-collecting set cover.
  \When{$(s_i, t_i)$ arrives}

  \State $d\gets$ the distance from $s_i$ to $t_i$ in $G/A$
  \State $j\gets \lfloor \log(d)+1\rfloor$\\
  \If {$\exists\ v\in \{s_i,t_i\}$ that is uncovered on layer $j$}
  \State Pass $r_{v, j}$ to $\pcsc$. \label{line:nmfl}
  \State If $r_{v,j}$ is covered by a newly added set $S_u$, add $u$ to $A$ and all $F_j$ with $c_u\geq \Csc$
  \label{line:add_open_facilities}\State For all $w\in B(v, \Rinc)$, set $y_{v,j}\gets y_{v,j} + 1$.

  \Else
  \State If $F_j\cap \bar{B}(s_i,\Rfac)$ is empty add $s_i$ to $F_j$. \label{line:extend_f_1}
  \State If $F_j\cap \bar{B}(t_i,\Rfac)$ is empty add $t_i$ to $F_j$.\label{line:extend_f_2}
  \State Let $v\in F_j\cap \bar{B}(s_i,\Rfac)$ and $w \in F_j\cap \bar{B}(t_i, \Rfac)$.
  \State Buy $v$ and $w$ and the shortest path between them in $G/A$. \label{line:buy_path_sf}
  \EndIf
  \EndWhen
  \State Rent the shortest path between $s_i$ and $t_i$ in $G/A$. \label{line:rent_rest} \label{line:greedy_path}
\end{algorithmic}
\end{algorithm}

\subsection{The analysis}\label{subsec:analysis}
In this section we prove Theorems~\ref{intro:main_thm_adaptive},~\ref{intro:main_thm_oblivious} and~\ref{intro:main_cor_det}. To do so, we are going to show that the cost of Algorithm~\ref{alg:nwrobsf} can be charged against the optimum cost of the instance of prize-collecting set cover.

We will consider the if-case and the else-case separately.
Let $\cif$ be the cost incurred in iterations that execute the if-case and $\celse$ be the cost incurred in the else-case. 
Let $\CSC$ be the cost of the constructed solution to the prize-collecting set cover instance and let $\CNIcon$ and $\CNIfac$ be the penalty costs and the costs for the added sets in the solution, respectively. 
Note that $\CSC = \CNIcon + \CNIfac$. We rely on the following four claims.

\begin{claim}
    $\celse \leq 6M\sum_{j}2^{j}|F_j|.$\label{claim:bound_celse}
\end{claim}
\begin{proof}
  The total cost incurred in an iteration in which the else-case is executed can be upper bounded by $3M\cdot 2^{j}$: at most $2M\cdot 2^{j}$ is incurred on \cref{line:buy_path_sf}, and at most $2^{j}$ is incurred on \cref{line:greedy_path}. So we can prove the above claim by simply bounding the number of iterations in which the else-case is executed for each $j$ by $2|F_j|$. Each such iterations falls into one of the following categories:
  \begin{itemize}
    \item If $s_i$ or $t_i$ are added to $F_j$ in Lines~\ref{line:extend_f_1} and~\ref{line:extend_f_2}, then $|F_j|$ increases by at least $1$ in this iteration. Clearly, this can happen at most $|F_j|$ times.
    \item Otherwise, some $v\in F_j\cap \bar{B}(s_i, \Rfac)$ and $w \in F_j\cap \bar{B}(t_i, \Rfac)$ are chosen. Now we claim that no path connecting $v$ and $w$ has been bought so far, i.e.\,$v$ and $w$ are not connected in $G/A$. If this was the case, then the distance $d$ between $s_i$ and $t_i$ in $G/A$ would be at most $\Rfac + \Rfac=2^{j-2}$, contradicting the fact that $j=\lfloor\log(d) +1\rfloor$\rnote{$\Rfacd \leq 2^{j-3}$}. Since $v$ and $w$ become connected in this iteration, the number of components in $A$ that contain a node in $F_j$ decreases by at least~$1$. Since the number of such components is at most $|F_j|$, this can happen at most~$|F_j|$ times.
  \end{itemize}
  Since we the number of iterations for each of the categories is at most $|F_j|$, the total number of iterations in which the else-case is executed for each $j$ is at most $2|F_j|$. This implies the claim.
\end{proof}

\begin{claim}
  $M\sum_{j}2^{j}|F_j|\leq 128(\CNIcon+\frac{\CNIfac}{\log \ktild}).$\label{claim:bound_Fsum}
\end{claim}
\begin{proof}
  Let $F_{j,1}$ be the set of nodes in $F_j$ that were added on Line~\ref{line:add_open_facilities} and let $F_{j,2}$ be the set of nodes in $F_j$ that were added on Lines~\ref{line:extend_f_1} and~\ref{line:extend_f_2}. Note that $F_j = F_{j,1}\cup F_{j,2}$.
  We will show that for each time that $|F_{j,1}|$ or $|F_{j,2}|$ increases by $1$, the right-hand side of the inequality increases by at least $M 2^{j}$. 
  \begin{itemize}
    \item When some new set $S_u$ for $u\in R_{v,j}$ is added to the set cover solution, $|F_{j,1}|$ increases by one for all $j$ with $c_v\geq \Csc$. Hence, the left-hand side of our claimed inequality grows by at most $M \sum_{i=0}^\infty 2^{6-i}\cdot c_v\leq 128M\cdot c_v$, while the right-hand side increases by $128M\cdot c_v$.

    \item When a node $v$ is added to $F_{j,2}$ (on Lines~\ref{line:extend_f_1} or~\ref{line:extend_f_2}), node $v$ can not have been uncovered on layer $j$. This means that either there exists a $z\in R_{v,j}\cap F_j\neq \emptyset$ or there is a $w\in B(v, \Rinc):\ y_{w,j}\geq M$.  However, in the former case, $z\in F_j \cap \bar{B}(v, \Rsc)$, in which case $v$ would not be added to $F_j$. So, there must be some $w\in B(v, \Rinc)$ such that $y_{w,j}\geq M$. Call $y_{w,j}$ a \emph{witness} for $v$.
    
    We will now charge the increase in $|F_{j,2}|$ to all iterations in which $y_{w,j}$ was increased. Note that in each such iteration, either a set of cost at least $\log(k)M\cdot\Csc$ is added to the set cover solution or a penalty of $2^{j}$ is paid. So in $M$ such iterations, the right-hand side of our claimed inequality increases by at least $M\cdot2^{j}$. 
    
    However, we need to be careful that an iteration in which $y_{w,j}$ is increased is not charged by multiple nodes $v$. We will show that this is not the case. 
    Consider an iteration in which $y_{w,j}$ is increased for all $w\in B(v, \Rinc)$. Suppose that $y_{w_1,j}$ and $y_{w_2,j}$ for $w_1,w_2\in B(v, \Rinc)$ are witnesses for $v$ and $v'$ respectively with $v\neq v'$. 
    Assume that $v$ gets added to $F_j$ before $v'$. 
    By the definition of the witnesses, we have $d(w_1, v) \leq \Rinc$ and $d(w_2, v') \leq \Rinc$. However, since $y_{w_1,j}$ and $y_{w_2,j}$ are both increased in the same iteration, we must have $d(w_1, w_2) \leq 2\cdot \Rinc$. 
    This implies that $d(v, v')\leq d(v, w_1) + d(w_1, w_2) + d(w_2, v')  \leq 4\cdot \Rinc=2^{j-3}$, which contradicts the fact that $v'$ is added to $F_j$ even though $v$ is already in $F_j$. Hence, each iteration in which $y_{w,j}$ is increased is charged by at most one node $v$. \rnote{$4\Rincd \leq \Rfacd$}
  \end{itemize}
\end{proof}

\begin{claim}
  $\cif\leq O(\CNIcon+\frac{\CNIfac}{\log \ktild}).$\label{claim:bound_cif}
\end{claim}
\begin{proof}
  We will show that each iteration the right-hand side increases faster than the left-hand side.
  Whenever the if-case is executed there are two cases:
  \begin{itemize}
    \item For iterations in which new set $S_v$ is added to the set cover solution, $\CNIfac$ increases by at least \smash{$\log \ktild \cdot M \cdot c_v\geq \log \ktild \cdot M \cdot \Csc$}. To buy $v$, a cost of $M\cdot c_v$ is incurred. A cost of at most $d\leq 2^j$ is incurred on Line~\ref{line:rent_rest} for renting the path between $s_i$ and $t_i$. So, $\cif$ increases by at most $M\cdot c_v+2^{j}$, which matches the increase in \smash{$O(\frac{\CNIfac}{\log \ktild})$} on the right-hand-side.
    \item In all other iterations in which the if-case is executed, $r_{v,j}$ is not covered and a penalty of $2^{j}$ is incurred in $\CNIcon$. The only cost incurred in $\cif$ is for renting a path between $s_i$ and $t_i$ on Line~\ref{line:rent_rest} of cost at most $2^j$. 
  \end{itemize}\end{proof}

\begin{claim}
    $\opt_{{\SC}} \leq O(\log \ktild) \cdot \opt_{\mathsf{RoB}}$.\label{claim:bound_opt}
\end{claim}
\begin{proof}
Consider $\opt$, an optimal solution to the given Rent-Or-Buy problem, and denote $A$ to be the set of nodes that are bought in $\opt$. We will construct a solution to prize-collecting set cover of cost at most $O(\log \ktild)$ times the cost of this optimal solution. To do so, we let the set cover consist of all sets $s_v$ corresponding to nodes $v$ that are bought in $\opt$. This incurs a set cost of $\sum_{v\in A} M\log(\ktild)c_v\leq O(\log \ktild)\cdot \opt$. It remains to show the same bound on the penalty costs. 

Consider such an arriving element $r_{v,j}$ that is not covered by the set cover. Let $v_1,v_2,\ldots, v_r$ be a sequence of nodes that connect $v=v_1$ to a node in $v_r\in \bd B(v, \Rsc)$ in the optimal Rent-Or-Buy solution. Suppose that $c_{v_r}\geq \Csc$. Then we must have $v_r\notin A$, since otherwise, we would have added $s_{v_r}$ in our set cover solution and $r_{v,j}$ would have been covered. So, renting $v_r$ incurred a cost of at least $\Csc$ in $\opt_\mathsf{RoB}$. So we can charge the penalty of $2^j$ incurred in $\opt_\SC$ to this cost.
    
On the other hand, if $c_{v_r}<\Csc$, then we will charge \smash{$\frac{c_{v_i}}{\log \ktild}$} of the penalty incurred by the arrival of $r_{v,j}$ to the cost paid for buying or renting $v_i$ in the optimal Rent-Or-Buy solution for $i\in \{2, r-1\}$. So, in this case we are charging \smash{$\sum_{i=2}^{r-1}\frac{c_{v_i}}{\log \ktild} \geq \frac{\Rsc -\Csc}{\log \ktild}= \frac{2^{j-6}}{\log \ktild}$}. This is indeed only a factor of $2^6\log \ktild$ less than the connection cost of $r_{v,j}$.\rnote{$\Cscd \leq \frac12 \Rscd$} \rnote{$\Cscd \leq \Rincd$} 

The only thing to check is that the cost for buying or renting a node in the optimal Rent-Or-Buy solution is larger than the charged cost. Nodes that are not bought in the optimal solution have to be rented at cost $c_v$ each time they are used. So this rental cost is clearly at least as large as the charged cost. 

For a node $v$ that is bought in the optimal rent-or-buy solution, observe that $y_{v,j}$ is increased for some $j$ each time that $\frac{c_{v_i}}{\log \ktild}$ is charged to it. Note that $y_{v,j}$ never becomes larger than~$M$. So, the total charged cost for $v$ is at most $M\frac{c_v}{\log \ktild}$ for each $j$. Hence, the total charged cost for $v$ is at most $M\cdot c_v$, which is exactly the cost of buying $v$ in the optimal solution.
\end{proof}

\begin{proof}[Proof of Theorems~\ref{intro:main_thm_adaptive},~\ref{intro:main_thm_oblivious} and~\ref{intro:main_cor_det}]
    By Claims~\ref{claim:bound_celse},~\ref{claim:bound_Fsum} and~\ref{claim:bound_cif}, the cost of incurred by online algorithm is at most
    \smash{$O(\CNIcon+\frac{\CNIfac}{\log \ktild})$}. By Claim~\ref{claim:bound_opt} we know that the optimal solution to the prize-collecting set cover instance is at most $O(\log \ktild)$ times the optimal solution to the Rent-Or-Buy Steiner Forest problem. 

    Plugging in the algorithms for prize-collecting set cover from Theorem~\ref{thm:rand_prize_collecting} in our algorithm, gives us $\CSCpen \leq O(\log|\mathcal{S}|)\opt_{{\SC}} $ and $\CSCsets \leq O(\log|\hat{X}|\log|\mathcal{S}|)\opt_{{\SC}}$. Hence, the cost of the online algorithm is at most:
    \begin{align*}
      \CNIcon+\frac{\CNIfac}{\log \ktild}\leq O\left(\log|\mathcal{S}|+\frac{\log|\hat{X}|\log|\mathcal{S}|}{\log \ktild}\right)\cdot \opt_{{\SC}} \leq O(\log|\mathcal{S}|(\log \ktild + \log|\hat{X}|))\cdot \opt_\mathsf{RoB}.
    \end{align*} 
    Since $|\tilde{X}| = O(\ktild\log\ktild)$ (recall that we set $\ktild := |T|$ when we instantiate Algorithm~\ref{alg:nwrobsf} with the deterministic algorithm from Theorem~\ref{thm:rand_prize_collecting}) and $|\S| = \bar{n}$, Theorems~\ref{intro:main_thm_adaptive} and~\ref{intro:main_thm_oblivious} follow. 
    Similarly, plugging in the deterministic algorithm from \cref{thm:det_prize_collecting} gives us a competitive ratio of $O(\bar{n}\log \ktild)$. This proves \cref{intro:main_cor_det}.
\end{proof}

\bibliographystyle{alpha}
\bibliography{ref}
\appendix
\section{Distances are polynomially related}\label{app:distreduction}

\begin{algorithm}
  \caption{Greedy\label{alg:offgreedy}}
\begin{algorithmic}[1]
  \State \textbf{Input:} Graph $G = (V, E)$, terminal pairs $(s_1, t_1), \ldots, (s_k, t_k)$.
  
\State \textbf{Output:} $\tilde{k}$-approximation to $\nwsf$ for rent-or-buy.
  \State Group terminals into $\tilde{k}$ distinct terminal pairs $(s_1, t_1), \ldots, (s_{\tilde{k}} , t_{\tilde{k}})$ and multiplicities $m_1, \ldots, m_{\tilde{k}}$.
  \For{$i \in \{1, 2, \ldots, \tilde{k}\}$}
  \State Compute shortest (cheapest) path between $s_i$ and $t_i$, $\mathbf{path} := \{s_i, v_{i_1}, \ldots, v_{i_\ell}, t_i\}$.
  \If{$m_i < p$}
  \State Rent $\mathbf{path}$.
  
  \Else
  \State Buy $\mathbf{path}$.
  \EndIf
  \EndFor
\end{algorithmic}
\end{algorithm}

Throughout this part, we make the assumption that $\tilde{k}$, the number of eventually arriving \emph{distinct} terminal pairs, is known (in the deterministic setting, just set it to $|T|$). This is without loss of generality and follows from a standard doubling argument, see for instance Appendix A of~\cite{BorstEliasVenzin25}. 

We now show as to why we can assume that all distances are between $1$ and $\tilde{k}^6$.

Throughout the algorithm, we maintain a \emph{guess} on the optimum value of an optimum value, $\mathsf{Guess}_\opt$, that we periodically update. Throughout (a run of) the algorithm, we maintain the property that
\begin{alignat}{1}\tag{P}
    \opt \leq \guess \leq \tilde{k}^2\cdot\opt.\label{property:guess}
\end{alignat}
Here, we slightly abuse notation and denote $\opt$ to be the optimum value of terminal pairs $(s_1, t_1)$ up to $(s_\ell, t_\ell)$, where $(s_{\ell+1}, t_{\ell+1})$ is the (next) pair that causes to update our guess. Specifically, if $(s_{\ell+1}, t_{\ell+1})$ is the first arriving terminal pair such that 
$$\greedy((s_1, t_1), \ldots, (s_{\ell+1}, t_{\ell+1})) > \guess,$$
we set 
$$\guess \leftarrow \tilde{k}\cdot\greedy((s_1, t_1), \ldots, (s_{\ell+1}, t_{\ell+1})),$$ 
and we rerun the algorithm from scratch - i.e.\,we forget all paths we already bought and we rerun the algorithm (in an online fashion) starting from terminal pair $(s_1, t_1)$.

During each run of the algorithm (e.g.\,for each value of $\guess$) we proceed as follows. For each arriving terminal pair $(s_i, t_i)$, we compute $d := d_{G/A}(s_i, t_i)$, where $A$ is the set of already bought nodes. Only if $\guess / (\tilde{k}^3 \cdot p) \leq d \leq 2\cdot \tilde{k}^2 \cdot \guess / p$, we pass the pair to Algorithm~\ref{alg:nwrobsf}. Observe that, up to scaling, all distances are polynomially related in $\tilde{k}$. Whenever $d < \guess / (\tilde{k}^4 \cdot p)$, we simply buy the cheapest path, and, whenever it exceeds $2\cdot\tilde{k}^2 \cdot \guess / p$ we rent the cheapest path. We will show that in both cases, we can charge the cost directly to $\opt$. We resume the procedure in the Algorithm~\ref{alg:dist_poly_related}.

\begin{algorithm}
  \caption{Distances are polynomially related\label{alg:procedure_polyk}}
\begin{algorithmic}[1]
  \State \textbf{Input:} Graph $G = (V, E)$. Terminal pairs $(s_1, t_1), \ldots, (s_k, t_k)$ arriving one-by-one.
  \State \textbf{Assumption:} $\tilde{k}$ (an approximation to) the number of distinct arriving terminal pairs.
  \State Initialize $\guess \leftarrow 0.$
  \State Initialize $\nwsf$, an instance of Algorithm~\ref{alg:nwrobsf}.
  \When{$(s_i, t_i)$ arrives}
  \If{$\guess < \tilde{k}\cdot \greedy((s_1, t_1), \ldots (s_i, t_i))$}
  \State $\guess \leftarrow \tilde{k}\cdot \greedy((s_1, t_1), \ldots (s_i, t_i))$.
  \State Reinitialize $\nwsf$.
  \EndIf
  \If{$d_G(s_i, t_i) < \guess / (p\cdot \tilde{k}^3)$}
  \State Buy the cheapest path from $s_i$ to $t_i$.~\label{line:k_buy}
  \ElsIf{$d_G(s_i, t_i) > 2\cdot\tilde{k}^3\cdot\guess / p$}
  \State Rent the cheapest path from $s_i$ to $t_i$.~\label{line:k_rent}
  \Else
  \State Pass $(s_i, t_i)$ to $\nwsf$.\label{line:pass_to_nwsf}
  \EndIf
  \EndWhen
\end{algorithmic}\label{alg:dist_poly_related}
\end{algorithm}

We now show that this procedure is correct.

\begin{claim}
    $\greedy$ is a $\tilde{k}$-approximation and can be assumed to be monotonically increasing.
\end{claim}
\begin{proof}
    For any $\tilde{k}$ distinct terminal pairs $(s_1, t_1), \ldots, (s_{\tilde{k}}, t_{\tilde{k}})$ with respective multiplicities $m_1, \ldots, m_{\tilde{k}}$, 
    $$\opt \geq \max_{i\in [\tilde{k}]}\opt( \overbrace{(s_i, t_i), \ldots, (s_i, t_i)}^{m_i}).$$
    Since the greedy approach is optimal when all terminal pairs are equal, $\greedy$ is a $\tilde{k}$-approximation. Note that for this $\tilde{k}$-approximation we do not require that any terminal pair share (bought) vertices, hence, by duplicating vertices, we may assume that the cost of the solution is monotonically increasing.
    
\end{proof}

\begin{claim}\label{claim:property_is_maintained}
    Property~\ref{property:guess} is maintained throughout.
\end{claim}
\begin{proof}
Let $\ell_1+1, \ell_2+1 \in \mathbb{N}$ two consecutive times where we update $\guess$, denote by $\guess$ the value used during this interval, and denote by $\opt_{\leq i}$ the optimum value to $\nwsf$ for terminal pairs $(s_1, t_1), \ldots, (s_{i}, t_{i})$. 

Since the greedy algorithm is a $\tilde{k}$-approximation and our update rule, at time $\ell_1+1$, it holds that 
$$\guess \leq \tilde{k}^2\cdot \opt_{\ell_1+1}.$$ 
Since $\opt$ can only increase, this shows the right-hand-side of the inequality.

On the other hand, since we do not update until time $\ell_2+1$, it must hold that 
$$\opt_{\leq \ell_2} \leq \greedy(\{(s_i, t_i)_{i\leq \ell_2}) < \guess.$$
\end{proof}

\begin{claim}\label{claim:double_every_three}
    After three updates of $\guess$, $\opt$ has increased by at least a factor $\tilde{k}$.
\end{claim}
\begin{proof}
    Note that every time we increase $\guess$, we increase it by at least a factor $\tilde{k}$. Since by property~$(\ref{property:guess})$, $\guess$ is between $\opt$ and $\tilde{k}^2\cdot \opt$, this means that $\opt$ must have increased by a factor~$\tilde{k}$.
\end{proof}

\begin{claim}\label{claim:extra_cost}
    For each value of $\guess$, the cost incurred in Lines~\ref{line:k_buy} and~\ref{line:k_rent} of Algorithm~\ref{alg:procedure_polyk} is at most $\guess/\tilde{k}^2$.
\end{claim}
\begin{proof}
    Whenever $d_{G/S}(s_i, t_i) < \guess/(p\cdot \tilde{k}^3)$, we buy the cheapest path between $s_i$ and $t_i$. This can only happen $\tilde{k}$ times and has overall cost $\guess/\tilde{k}^2$.
    On the other hand, whenever $d_{G/S}(s_i, t_i) > 2\cdot\tilde{k}^3 \cdot \guess / p$, the cost of $\greedy$ upon arrival of $(s_i, t_i)$ must increase by at least $\tilde{k}^3\cdot\guess / p$. Indeed, otherwise, $\greedy$ must have already bought nodes worth at least  $p\cdot d_{G/S}(s_i, t_i)/2 > \tilde{k}^3 \cdot \guess$ prior to the arrival of $(s_i, t_i)$, which would have caused to update $\guess$. Hence, until we update $\guess$, the total cost incurred on Line~\ref{line:k_rent} is at most $\guess / \tilde{k}^2$.
\end{proof}

\begin{theorem} Algorithm~\ref{alg:procedure_polyk} is $O(\log \tilde{k} \log n)$-competitive for $\nwsf$.
\end{theorem}

\begin{proof}
    By Claim~\ref{claim:double_every_three}, denoting by $\ell_1, \ell_2, \ldots, \ell_{\text{fin}} \in \{1, \ldots, k\}$ all indices causing $\guess$ to be updated, we have that
    $$\opt_{\leq \ell_n} \geq \tilde{k}\cdot \opt_{\leq \ell_{n-3}}.$$
    Since Algorithm~\ref{alg:nwrobsf} is $O(\log \tilde{k} \log n)$-competitive for each value of $\guess$, we have that the total cost incurred in Line~\ref{line:pass_to_nwsf} (i.e.\,the cost from the pairs we actually pass to the algorithm for $\nwsf$) is at most
    $$3 \cdot \sum_{i \geq 1} \frac{1}{\tilde{k}^i} \cdot O(\log \tilde{k}\cdot\log n)\cdot \opt_{\leq k}.$$
    Since the sum converges, it remains to bound the total cost incurred in Lines~\ref{line:k_buy} and~\ref{line:k_rent}. By Claim~\ref{claim:extra_cost}, for each value of $\guess$, this is at most $2\cdot \guess / \tilde{k}^2$. By property~\ref{property:guess} which is maintained throughout, see Claim~\ref{claim:property_is_maintained}, this cost is at most twice the optimum value of servicing all terminal pairs until we reupdate $\guess$. Summing over all runs and using again Claim~\ref{claim:double_every_three}, we can bound this by
    $$2\cdot \sum_{i \geq 1} \frac{3}{\tilde{k}^i} \cdot \opt_{\leq k}.$$
    This concludes the proof.
\end{proof}

\section{Applying the result from \cite{AwerbuchAzarBartal96}}
\label{app:bbreduction}
The result given in \cite{AwerbuchAzarBartal96} is phrased in terms of \emph{task systems}. We repeat the relevant definitions here. A task system is a tuple $\mathcal{P}=(\TaskSpace, \Conf, \task, \mcost)$. Here $\TaskSpace$ is a set of \emph{tasks}, $\Conf$ is a set of \emph{configurations}, and $\task(C, t)$ is the cost of executing task $t$ in configuration $C$. The cost of moving from configuration $C$ to $C'$ is denoted by $\mcost(C, C')$.

A sequence of tasks arrives one by one. An online algorithm can move to a new configuration ~$C$ on arrival of each task and will then have to pay the cost of changing the configuration and executing the task in configuration~$C$.

A task system is said to be \emph{forcing} if $\task(C, t)$ is either $0$ or $\infty$ for all $C\in \Conf$ and $t\in \TaskSpace$.
From a forcing task system $(\TaskSpace, \Conf, \task, \mcost)$ an $M$-relaxed task system $(\TaskSpace, \Conf, \task', \mcost')$ can be obtained by setting $\mcost(C, C') = M\cdot \mcost'(C, C')$ for all $C, C'\in \Conf$ and $\task'(C, t) = \min_{C': \task(C', t) =0} \mcost(C, C')$ for $M\geq 1$.

In \cite{AwerbuchAzarBartal96} it is shown that for a forcing task system that has an algorithm $\mathcal{A}$ that is $c$-competitive against adaptive online adversaries, the following algorithm, called $\mathcal{A}'$, is $3c$-competitive for the $M$-relaxed task system against adaptive adversaries. The algorithm $\mathcal{A}'$ is very simple: for each arriving task $t$, pass it to algorithm $\mathcal{A}$ with probability $1/(2M)$ and move to the same configuration that $\mathcal{A}$ moves to. Otherwise, stay in the current configuration and pay the cost of serving the task in the current configuration.

Steiner forest problems in the buy-only setting can be formulated as forcing task systems. Configurations correspond to sets of bought edges and nodes, and tasks correspond to arriving requests. 
The Rent-Or-Buy setting with parameter $M$ then directly corresponds to a $M$-relaxed task system: the cost of moving moving from one configuration to another (i.e.\,by buying the missing edges and nodes) is $M$ times the cost of servicing the request in the current configuration (i.e.\,renting the missing edges and nodes). Hence, instantiating this procedure with an online algorithm that holds against adaptive adversaries, for instance the $O(\log^2 n)$-competitive deterministic algorithm for online node-weighted Steiner forest from~\cite{BorstEliasVenzin25}, directly gives the same guarantees in the rent-or-buy setting. 
However, there are some subtleties when one tries to apply this reduction to the $O(\log k \log n)$-competitive algorithm from~\cite{BorstEliasVenzin25}. 
This is because the guarantees of this (randomized) algorithm are only stated to hold against an oblivious adversary, a weaker adversarial setting than that of an adaptive online adversary. Still, it is possible to use their algorithm to obtain a randomized $O(\log \ktild \log n)$-competitive algorithm in the rent-or-buy setting against oblivious adversaries. We discuss this in the remainder of this section.

However, if the task sequence to $\mathcal{A}'$ is guaranteed to consists of only tasks from some set $S\subseteq \TaskSpace$, then we can consider the modified task system $(S, \Conf, \task', \mcost')$. If $\mathcal{A}$ is $c$-competitive for the original forcing task system against any adversary that only sends requests for tasks from $S$, then $\mathcal{A}'$ is $c$-competitive against adaptive adversaries on the modified task system. Hence, by the reduction from \cite{AwerbuchAzarBartal96}, $\mathcal{A}'$ is $3c$-competitive for the $p$-relaxed task system against adaptive adversaries on this modified task system. However, the behavior of algorithm $\mathcal{A}'$ does only depend on the tasks that it actually receives. So, we can say that $\mathcal{A}'$ is $3c$-competitive for the $p$-relaxed task system against adaptive adversaries that only send requests for tasks from $S$.
\begin{corollary}
  If $\mathcal{A}$ is $c$-competitive for a forcing task system against any adversary that only sends requests for tasks from a set $S$, then algorithm $\mathcal{A}'$ is $3c$-competitive for the $p$-relaxed task system against any such adversary.
\end{corollary}
We now note that the randomized algorithm is $O(\log K \log n)$-competitive against a semi-adaptive adversary that commits to a set of at most $K$ potential requests, as the proof only uses the fact that all requests come from this set. Hence, this induces a randomized $O(\log K \log n)$-competitive algorithm for the rent-or-buy problem against a semi-adaptive adversary.

\end{document}